\documentclass[letterpaper, 11pt]{amsart}
\usepackage{macros}
\usepackage{fullpage}
\usepackage[foot]{amsaddr}

\allowdisplaybreaks

\def\BIP{\mathsf{bip}}
\def\CM{\mathsf{CM}}
\def\GD{\mathsf{GD}}
\def\LOW{\mathsf{Low}}
\def\PATTERN{{\mathscr{P}}}
\def\RANK{\mathsf{rank}}

\def\Ber{\mathrm{Ber}}
\def\ParameterCondition{{\Delta (1 - \e^{-\beta}) \ge 9 q \ln{\frac{q \Delta}{1 - \e^{-\beta}}}}}

\def\e{\mathrm{e}} 
\def\mix{\mathrm{mix}}

\def\vecpi{\boldsymbol{\pi}}

\newboolean{doubleblind}
\setboolean{doubleblind}{False}

\title{Counting and Sampling Anti-ferromagnetic Potts Models on Random Regular Bipartite Graphs in the Non-uniqueness Regime}

\ifdoubleblind
    \author{Author(s)}
\else
    \author{Zhidan Li, Siyu Liu, Kuan Yang}
    \address[Zhidan Li]{School of Computer Science, Shanghai Jiao Tong University, Shanghai, China. \textnormal{Email: \url{yueEnTeRnAL@sjtu.edu.cn}}.}
    \address[Siyu Liu]{Zhiyuan College, Shanghai Jiao Tong University, Shanghai, China. \textnormal{Email: \url{williamliusy@sjtu.edu.cn}}.}
    \address[Kuan Yang]{John Hopcroft Center for Computer Science, Shanghai Jiao Tong University, Shanghai, China. \textnormal{Email: \url{kuan.yang@sjtu.edu.cn}}.}
\fi
\date{}

\begin{document}

\begin{abstract}
The anti-ferromagnetic multi-state Potts model, a generalization of the Ising model, is one of the most fundamental models in statistical physics.
It was conjectured by Koteck\'{y} (Phys. Rev. B, 1985) that the model undergoes a phase transition from a disordered phase at infinite temperature to an ordered phase at sufficiently low temperature on lattices.
Such phase transitions are believed to play an important role in computational complexity theory and remain closely connected to the problem of approximating the partition function of the system.
For proper three-coloring models (corresponding to the zero-temperature), torpid mixing of a family of local-update Markov chains on lattices was established by Galvin, Kahn, Randall and Sorkin (SIDMA, 2015), coinciding with the presence of phase coexistence following shown by Feldheim and Spinka (J. Eur. Math. Soc., 2019).

In this work, we study approximating the partition function of the anti-ferromagnetic multi-state Potts model at low temperature on random regular bipartite graphs, which are with high probability good bipartite expanders.
On the negative side, we generalize the result by Geisler, Kang, Sarantis and Wdowinski (arXiv, 2026) for anti-ferromagnetic Ising models to show that when the temperature is sufficiently low relative to the degree of the underlying graph, the celebrated single-site Glauber dynamics has exponentially slow mixing time.
On the positive side, we design a deterministic algorithm that yields an approximation to the partition function of the model via the framework of abstract polymer models as Jenssen, Keevash and Perkins (SICOMP, 2020), Liao, Lin, Lu and Mao (Theor. Comput. Sci., 2022), Galanis, Goldberg and Stewart (TOCT, 2021) and Geisler, Kang, Sarantis and Wdowinski (arXiv, 2026).
\end{abstract}

\maketitle

\newpage

\section{Introduction} \label{sec:introduction}

The $q$-state Potts model, introduced by R.~B.~Potts~\cite{Potts52}, is a fundamental spin system in statistical physics and has been well studied in probability theory and theoretical computer science. It generalizes the Ising model by allowing each site to take one of $q$ possible colors.
Formally, given a ground set $U$, the model is specified by a color set $[q] \defeq \set{1, \ldots, q}$ for a positive integer $q \ge 3$, an inverse temperature $\beta$ and a Hamiltonian $H : [q]^U \to \mathbb R$.
A \emph{configuration} is an assignment $\sigma : U \to [q]$ and the weight of $\sigma$ is defined by:
\begin{align} \label{eq:weight-function}
    w(\sigma) \defeq \e^{-\beta H(\sigma)}.
\end{align}
Let the state space or configuration space $\Omega \defeq [q]^U$ be the collection of all configurations.
When $\beta < 0$, we say that it is a \emph{ferromagnetic $q$-Potts model} and an \emph{anti-ferromagnetic} case for $\beta > 0$.

We study the \emph{anti-ferromagnetic $q$-Potts model} on graphs.
Namely, for a graph $G = (V, E)$, the ground set is $V$ (and thus $\Omega = [q]^V$), and the Hamiltonian $H(\cdot)$ is given by the number of monochromatic edges in $G$ under a configuration, \IE,
\begin{align} \label{def:Hamiltonian-function}
    \forall \sigma \in \Omega, \quad H(\sigma) = m_G(\sigma) \defeq \abs{\set{(u, v) \in E \midc \sigma|_u = \sigma|_v}}.
\end{align}
The Gibbs distribution induced by the anti-ferromagnetic $q$-Potts model on $G$ at $\beta$ is the probability distribution $\mu = \mu_{G, q; \beta}$ on $\Omega$ defined as
\begin{align*}
    \forall \sigma \in \Omega, \quad \mu(\sigma) = \frac{w(\sigma)}{Z_{G, q}(\beta)} = \frac{\e^{-\beta m_G(\sigma)}}{Z_{G, q}(\beta)}
\end{align*}
where $Z_{G, q}(\beta) = \sum_{\sigma \in \Omega} w(\sigma)$ is the \emph{partition function}.
In particular, when the temperature is zero ($\beta = \infty$), configurations with monochromatic edges receive weight $0$ and configurations representing \emph{proper colorings} receive weight $1$.
Thus the proper vertex-coloring model is a special case of the anti-ferromagnetic Potts model.

A central computational problem on the anti-ferromagnetic Potts model is to compute the partition function $Z_{G, q}(\beta)$.
Unfortunately, exact computation is $\numberP$-complete, even on regular and triangle-free graphs \cite{Greenhill00}.
This motivates the following natural algorithmic question:
\begin{center}
    \emph{Does the anti-ferromagnetic Potts partition function admit an efficient deterministic or
randomized approximation scheme?}
\end{center}
Here we use $\FPTAS$ and $\FPRAS$ to describe approximation algorithms for counting problems: an $\eps$-approximation to the target quantity is produced in time polynomial in the input size and $\eps^{-1}$, deterministically for an $\FPTAS$ and with constant success probability for an $\FPRAS$.\footnote{Formally, for a function $Z$ from a family of instances to reals, a \emph{fully polynomial-time approximation scheme} ($\FPTAS$) is an algorithm such that for each instance $\Phi$ and a tolerance error $\eps \in (0, 1)$, it outputs a number $\wh{Z}$ in time $\poly{\abs{\Phi}, \eps^{-1}}$ (where $\abs{\Phi}$ is the length of its input) satisfying that $\wh{Z}$ is an $\eps$-approximation to $Z(\Phi)$, \IE, $(1 - \eps)Z(\Phi) \le \wh{Z} \le (1 + \eps) Z(\Phi)$.
Similarly, a \emph{fully polynomial-time randomized approximation scheme} ($\FPRAS$) is a randomized algorithm outputting a $\wh{Z} \in (1 \pm \eps)Z(\Phi)$ in time $\poly{\abs{\Phi}, \eps^{-1}}$ with probability at least $3/4$. Here the probability $3/4$ is standard. One can substitute it with any real number $> 1/2$.}

A guiding principle in approximate counting is the connection between efficient algorithms and the \emph{uniqueness} of the infinite-volume Gibbs measure.
For the anti-ferromagnetic Potts model on the infinite $\Delta$-regular tree, the conjectured uniqueness threshold is
\begin{align*}
    \beta_c(q, \Delta) \defeq \ln{\max\set{1, \frac{\Delta}{\Delta - q}}}.
\end{align*}
It is a folklore conjecture that the Gibbs distribution is unique if and only if $\beta \le \beta_c(q, \Delta)$ (equivalently, $q \ge \Delta(1 - \e^{-\beta})$).
This has been proved for $q=3$ by Galanis, Goldberg and Yang~\cite{GGY18} and subsequently improved by Bencs et al.~\cite{BBBG23} for general $q \ge 3$ but large enough $\Delta$.
In a recent work by Chen et al.~\cite{CLMM23}, the existence of tree-uniqueness is shown for $q \ge (1 - \e^{-\beta})(\Delta + 5) + 1$, which almost matches the conjectured threshold.
The tree-uniqueness threshold plays a crucial role in computational theory since it is conjectured that there exists an efficient approximation (either an $\FPTAS$ or $\FPRAS$) to the partition function of anti-ferromagnetic Potts models on general bounded degree graphs if and only if $\beta$ is in the tree-uniqueness threshold (unless $\mathbf{P}=\mathbf{NP}$), which has been proved in the anti-ferromagnetic Ising model (\cite{LLY13,GSV16}).

On the algorithmic side, Lu and Yin~\cite{LY13} propose an $\FPTAS$ via the so-called \emph{correlation-decay} method when $3\Delta(\e^{\abs{\beta}} - 1) < 1$ for general Potts models.
Correlation decay algorithms are based on the property of \emph{spatial mixing}, while a recently developed technique \emph{spectral independence} connects spatial mixing to rapid mixing of Glauber dynamics for sampling colorings on high-girth graphs~\cite{CLMM23}.
In parallel, zero-free methods give another route to deterministic counting algorithms for the Potts models, particularly the proper coloring problem. Liu, Sinclair and Srivastava~\cite{LSS25} establish the zero-freeness of $Z_{G, q}(\beta)$ when $q \ge 2\Delta$, leading to an $\FPTAS$ by~\cite{PR17}. The zero-freeness result is recently improved to $q \ge (2 - \eta)\Delta$ for some $\eta \ge 0.002$ by~\cite{BBR26}.
For random graphs, Blanca et al.~\cite{BGSVY20} designed an efficient sampler on random $\Delta$-regular graphs with polynomially small total variation error in the tree-uniqueness regime, although such a sampler does not directly yield an $\FPRAS$ for random regular graph instances.

The picture changes sharply in the non-uniqueness regime. On general graphs, Galanis, \v{S}tefankovi\v{c} and Vigoda~\cite{GSV15} proved inapproximability for the anti-ferromagnetic $q$-Potts partition function on triangle-free $\Delta$-regular graphs for even $q$, unless $\NP = \RP$.
On bipartite graphs, however, the obstruction takes a
different form. Low-temperature anti-ferromagnetic systems may exhibit several ordered phases
rather than a single disordered phase.
When the model is defined on lattice $\mathbb{Z}^{\Delta}$, \Kotecky~\cite{Kotecky85} conjectured that there exists a phenomenon called \emph{long-range order} for anti-ferromagnetic Potts models.
This was later proved by Feldheim and Spinka~\cite{FS19} that such coexistence of different phases occurs for $q = 3$ and sufficiently large $\Delta$ at $\beta > \ln{\Delta}$, which is far away from the tree-uniqueness threshold. In the zero-temperature case (where $\beta = \infty$), Galvin et al.~\cite{GKRS15} showed torpid mixing of any local-update chains for proper $3$-colorings on $\mathbb{Z}^{\Delta}$ on sufficiently high dimensions, where the model is in the ordered phase.

This phase-coexistence viewpoint has also led to positive algorithmic results on random bipartite graphs.
Jenssen, Keevash and Perkins~\cite{JKP19} gives an $\FPTAS$ for counting the number of proper colorings on random $\Delta$-regular bipartite graphs with high probability, as long as $\Delta = \Omega((q\ln{q})^2)$, by decomposing the state space into ordered phases and polymer models around these phases.
Similar techniques are extended to general spin systems (including both ferromagnetic and anti-ferromagnetic Potts models) by Galanis, Goldberg and Stewart~\cite{GGS21} to design an $\FPRAS$ to compute the models on random $\Delta$-regular bipartite graphs when $\Delta(1 - \e^{-\beta})^{4} = \Omega((q \ln{q})^4)$
and a following work by Chen et al.~\cite{CGSV22} improves the result for proper vertex-coloring models to $\Delta \ge 100q\ln{(q\Delta)}$.
These works suggest a compelling question:
\begin{center}
    \emph{Can one exploit ordered phases on random regular bipartite graphs to count anti-ferromagnetic Potts models at low temperature, even though local-update Markov chains mix slowly?}
\end{center}

\subsection{Main results}

We answer this question for the anti-ferromagnetic Potts model on random regular bipartite graphs at low temperature in the non-uniqueness regime.
Random regular bipartite graphs are good expanders with high probability~\cite{JKP19,GGS21,LLLM22,CGSV22}, and this expansion forces typical low-temperature configurations to concentrate near ordered coloring patterns.
Our results show both sides of the algorithmic story: the standard local Markov chain is torpid, but the partition function is still deterministically approximable through a phase-based polymer expansion.


Our first result is slow mixing of single-site Glauber dynamics.
The Glauber dynamics $P^{\GD}$ is one of the most celebrated Markov chains with local updates.
To sample from a target distribution $\pi$ on $\Omega$, starting from an initial state $X_0 \in \Omega$, at each step we update the current state $X_t$ as following:
\begin{itemize}
    \item pick a vertex $v \in V$ uniformly at random;
    \item update $X_{t + 1} \sim \pi(\cdot\;|\;X_{t + 1}(V \setminus \set{v}) = X_t(V \setminus \set{v}))$.
\end{itemize}
We show that it has an exponential mixing time on the anti-ferromagnetic $q$-Potts models on random bipartite regular graphs at sufficiently low temperature.
\begin{theorem} \label{thm:Glauber-dynamics-slow-mixing}
    For every positive integers $q, \ge 3, \Delta \ge 4$ and a real number $\beta > 0$ satisfying
    \begin{align*}
        \ParameterCondition,
    \end{align*}
    the following holds with high probability over $G$ drawn uniformly from all $\Delta$-regular bipartite $2n$-vertex graphs at random for all sufficiently large positive integer $n$.
    The Glauber dynamics $P^{\GD}$ for the anti-ferromagnetic $q$-Potts models on $G$ at inverse temperature $\beta$ has a mixing time at least $\e^{Cn}$ for some factor $C = C(q, \Delta, \beta)$.
\end{theorem}

Despite the slow mixing of Glauber dynamics (\Cref{thm:Glauber-dynamics-slow-mixing}), our second result gives an efficient deterministic approximation algorithm for the partition function in a low temperature range.
\begin{theorem} \label{thm:counting-results}
    For every positive integers $q \ge 3, \Delta \ge 4$ and a real number $\beta > 0$ satisfying
    \begin{align*}
        \ParameterCondition,
    \end{align*}
    the following holds with high probability over $G$ drawn uniformly from all $\Delta$-regular bipartite $2n$-vertex graphs at random for all sufficiently large positive integer $n$.
    There exists a deterministic algorithm on input $G$ and $\eps \in (0, 1)$ to output $\wh{Z}$ satisfying that
    \begin{align*}
        (1 - \eps) Z_{G, q}(\beta) \le \wh{Z} \le (1 + \eps) Z_{G, q}(\beta)
    \end{align*}
    in running time $\poly{n, \eps^{-1}}$, \IE, an $\FPTAS$ for counting the partition function of the anti-ferromagnetic $q$-Potts model on $G$ at $\beta$.
\end{theorem}
We mark here that the analysis of~\Cref{thm:counting-results} works for $\beta = \infty$, \IE, the proper vertex-coloring model on random regular bipartite graphs.
\begin{corollary} \label{cor:proper-coloring-result}
    For every positive integers $q \ge 3, \Delta \ge 4$ satisfying
    \begin{align*}
        \Delta \ge 9q\ln{(q\Delta)},
    \end{align*}
    with high probability over $G$ drawn uniformly from all $\Delta$-regular bipartite $2n$-vertex graphs at random for all sufficiently large positive integer $n$, there exists an $\FPTAS$ for the number of proper $q$-colorings of $G$.
\end{corollary}

\subsection{Technical overview}

We state here a brief overview of techniques used in our work, with some details in~\Cref{sec:preliminaries,sec:anti-ferromagnetic-Potts}.

\subsubsection*{\bf Existence of ordered phases}

To show~\Cref{thm:Glauber-dynamics-slow-mixing,thm:counting-results}, we investigate the geometric structure of the configuration space of the anti-ferromagnetic Potts model on random regular bipartite graphs.
A key observation is that, at low temperature, monochromatic edges are penalized, and expansion makes it costly for a color to appear substantially on both sides of the bipartition.
Consequently, typical configurations are close to one of the `coloring patterns', similar to the anti-ferromagnetic $3$-states case on the lattice $\mathbb{Z}^{\Delta}$ in~\cite{GKRS15,FS19}.
Based on this observation, we apply the concept of patterns in~\cite{JKP19}\footnote{We note here that the concept of patterns in~\cite{JKP19} is a special case of the maximal bi-cliques in~\cite{GGS21}.} to define a rank measuring the number of vertices that disagree with the nearest pattern.
The main probabilistic estimate shows that high-rank configurations have exponentially small Gibbs weight.

\subsubsection*{\bf Conductance}
For the slow mixing result, we use the standard conductance method of Sinclair and Jerrum~\cite{SJ89}.
By a type of the Cheeger's inequality, there exists a bridge between the conductance of a given Markov chain and the second largest eigenvalue of its Laplacian matrix, leading to both upper and lower bound to its mixing time.
For anti-ferromagnetic Ising models ($q = 2$ in Potts) on a family of regular bipartite expander graphs, a recent work by Geisler et al.~\cite{GKSW26} has shown that the conductance of Glauber dynamics is exponentially small, and hence the mixing time is exponential in the size of the instance. Similarly, to show~\Cref{thm:Glauber-dynamics-slow-mixing}, we apply the method of conductance. 
We show how to find a subset of the state space with small conductance and apply a simpler lower bound for the mixing time by Levin and Peres~\cite{LP17} than the one induced by the Cheeger's inequality.
To find the very subset, we also prove some useful results to describe the geometry of the anti-ferromagnetic $q$-Potts models on random regular bipartite graphs.

\subsubsection*{\bf Abstract polymer models}

The \emph{abstract polymer model} together with \emph{the cluster expansion} is a useful tool to study statistical mechanics models.
For approximate counting problems, the abstract polymer model has been successfully used to design $\FPTAS$ or $\FPRAS$ on expander graphs for the hard-core models \cite{JKP19,CGGPSV21,LLLM22,CGSV22}, ferromagnetic Potts models \cite{JKP19,CGGPSV21,GGS22,CGSV22}, general spin systems \cite{GGS21,CGSV22} and proper colorings \cite{JKP19,LLLM22,CGSV22}.

We follow the methodology in~\cite{JKP19,LLLM22,GGS21,CGSV22}.
Based on our description of ordered phases, we specify each configuration in $\Omega$ with patterns.
Therefore, for each pattern, we design a suitable abstract polymer model and follow a routine way to obtain an efficient approximation scheme to the partition function by verifying the \Kotecky-Preiss condition.

\subsection{Comparison with related works}

The torpid mixing of Glauber dynamics has been recently shown on the anti-ferromagnetic Ising models with uniform external fields by~\cite{GKSW26} on bipartite expander graphs.
Their result relies on an observation that the set of `balanced' configurations owns an exponentially small weight.
For the multi-state Potts model, the main additional challenge arises from the difficulty to identify an appropriate analogue of `balanced configurations' and upper bound their total weights.
To overcome this, we investigate the geometry of the configuration space and show that the configurations are highly concentrated around some specific `ground states'.
Therefore, we are able to define a proper `boundary' and show the small conductance of the Glauber dynamics.

For approximately counting, our work is close to the polymer algorithms for proper colorings by Jenssen, Keevash and Perkins~\cite{JKP19}, and for general spin systems by Galanis, Goldberg and Stewart~\cite{GGS21} and Chen et al.~\cite{CGSV22}. It has been shown in~\cite{JKP19} that the state space can be viewed as a union of configurations specified by different coloring `patterns' representing assigning almostly disjoint colors to each side for proper vertex-colorings.
The latter work~\cite{GGS21} focuses on soft-constrained models and applies maximal bi-cliques which is a generalization of patterns and has been successfully used to design an efficient sampler on the ferromagnetic case in~\cite{CGGPSV21}.
As a result,~\cite{GGS21} gives an $\FPRAS$ for counting the anti-ferromagnetic Potts models on random regular bipartite graphs, under a stronger constraint on the degree $\Delta$ and the number of colors $q$.
The following work~\cite{CGSV22} adapts the method of phase vectors and phase maximality in~\cite{GSV15,CGGPSV21} and improves the threshold for proper vertex-coloring models.
However, their method requires a strict condition on phase maximality, which is seemingly hard to verify for Potts models.
To overpass this obstacle, we use a more direct analysis for anti-ferromagnetic Potts models on random regular bipartite graphs and design a \emph{deterministic} algorithm whenever $\ParameterCondition$.

\subsection{Organization of this paper}

We present preliminaries including some notations, definitions and useful lemmas of graph-expansion and Markov chains in~\Cref{sec:preliminaries}.
Then in~\Cref{sec:anti-ferromagnetic-Potts}, we develop the rank and pattern decomposition and prove
the high-rank tail bound for random regular bipartite graphs.
The proof of slow mixing (\Cref{thm:Glauber-dynamics-slow-mixing}) is presented in~\Cref{sec:slow-mixing-theorem}. Our counting algorithm and the proof is given in~\Cref{sec:counting-theorem}.
Finally, in~\Cref{sec:conclusion}, we conclude our work with some open questions.

\subsection*{Statement on AI use}

The improvement of the tail bound~\Cref{lem:high-rank-tail-bound} is optimized under the help of ChatGPT 5.6 Pro Sol based on the original version which applies for bipartite expanders.
The choice of the parameter $\theta$ and related calculations are also derived with the help of ChatGPT 5.6 Pro Sol.

\section{Preliminaries} \label{sec:preliminaries}

\subsection{Mathematical notations}

We list notations used in the paper here.
We use $\e$ to denote the natural logarithm base and $\ln$ to denote the natural logarithm function.
For a natural number $n$, we use $[n]$ to denote the set $\set{1, 2, \ldots, n}$.
For two sets $A, B$, we use $A \sqcup B$ to denote the disjoint union of $A$ and $B$.

A graph $G$ consists of a finite vertex set $V$ and an edge set $E \subseteq V \times V$.
Given a graph $G = (V, E)$, we denote by $N_G(S)$ the neighborhood of $S$ for every $S \subseteq V$, \IE,
$$
    N_G(S) \defeq \set{v \in V \setminus S \mid \exists u \in S, (u, v) \in E}
$$
and for simplicity, use $N_G(v)$ to denote $N_G(\set{v})$ when $S = \set{v}$ for each $v \in V$.
For two vertex subsets $S, T \subseteq V$, we use $E_G(S, T) \defeq \set{(u, v) \in E \midc u \in S \land v \in T}$ to denote the edges between $S$ and $T$.
For a subset $S \subseteq V$ and an integer $\ell \ge 1$, we say $S$ is $\ell$-connected if the induced subgraph $G^{\ell}[S]$ is connected.
For simplicity, we use $G = (V = \calL \sqcup \calR, E)$ to denote a \emph{bipartite graph}, where $(\calL, \calR)$ forms a partition of the vertex set $V$, and the edge set $E \subseteq \calL \times \calR$.

For an assignment $\sigma : U \to [q]$ and a subset $S \subseteq U$, we use $\sigma(S)$ to denote the partial assignment of $\sigma$ projected on $S$, \IE, $\sigma|_S$ and use $\sigma(v)$ to denote $\sigma(\set{v})$ for simplicity when $S = \set{v}$.
With little abuse of notation, we also use $\sigma(v)$ to denote the color assigned to $v$ by $\sigma$.
For a color subset $C \subseteq [q]$, we denote by $\sigma^{-1}(C)$ the subset $S \subseteq U$ assigned colors in $C$ by $\sigma$, \IE, $\sigma^{-1}(C) \defeq \set{v \in U \mid \sigma(v) \in C}$.

For a state space $[q]^U$, we use $d_H(\cdot, \cdot)$ to denote the Hamming distance on it, \IE,
\begin{align*}
    \forall \sigma, \tau \in [q]^U, \quad d_H(\sigma, \tau) = \sum_{v \in U} \id{\sigma(v) \neq \tau(v)}.
\end{align*}
We extend $d_H(\cdot, \cdot)$ to subsets of $[q]^U$ as
\begin{align*}
    \forall \Lambda, \Pi \subseteq [q]^{U}, \quad d_H(\Lambda, \Pi) = \min_{\sigma \in \Lambda, \tau \in \Pi} d_H(\sigma, \tau).
\end{align*}

Fix a positive integer $d \ge 2$.
The probability simplex of dimension $d$ contains all $d$-dimensional probability vectors, \IE,
\begin{align*}
    \Delta_d \defeq \set{\vecx \in \mathbb{R}^d \mid \forall i \in [d], x_i \ge 0 \land \sum_{i = 1}^d x_i = 1}.
\end{align*}
We define the entropy $H$ of $\vecx$ as $H(\vecx) = \sum_{i = 1}^d x_i \ln x_i$ with convention $0 \ln 0 = 0$.
With little abuse of notations, the entropy function $H : [0, 1] \to \mathbb{R}$ is defined as $H(x) = -x\ln{x} - (1 - x)\ln{(1 - x)}$ with $H(0) = H(1) = 0$.

\subsection{Random regular bipartite graphs} \label{sec:random-regular-bipartite-graph}

Now we discuss here some combinatorial properties of random regular bipartite graphs.
For positive integers $\Delta, n \in \mathbb N_{+}$, we use $\calG_{n, \Delta}^{\BIP}$ to denote the probabilistic model of random $\Delta$-regular bipartite graphs of $2n$ vertices.

To study random regular bipartite graphs $\calG_{n, \Delta}^{\BIP}$, we introduce another classical probabilistic model of random graphs named the \emph{configuration model}.
\begin{definition}[Configuration model] \label{def:configuration-model}
    Let $G = (V = \calL \sqcup \calR, E) \sim \calG_{n, \Delta}^{\CM}$ be a bipartite multi-graph generated by the following procedure:
    \begin{itemize}
        \item The vertex set $V = \calL \sqcup \calR$ is set by $\calL = [n]$ and $\calR = [n]$.
        \item To generate $E$, we sample $\Delta$ perfect matchings of the complete bipartite graph $K_{n, n} = (\calL \sqcup \calR, \calL \times \calR)$ uniformly at random and for each perfect matching, we add edges in it to $E$. 
    \end{itemize}
\end{definition}
The following proposition builds a connection between $\calG_{n, \Delta}^{\CM}$ and $\calG_{n, \Delta}^{\BIP}$.
\begin{proposition}[\cite{Wormald99}] \label{prop:configuration-model-to-random-regular-bipartite-graph}
    For an event $\calE$, if $\Pr[G \sim \calG_{n, \Delta}^{\CM}]{G \in \calE} = o(1)$, then $\Pr[G \sim \calG_{n, \Delta}^{\BIP}]{G \in \calE} = o(1)$.
\end{proposition}

An important property of random regular bipartite graphs is the \emph{vertex expansion}.
For a graph $G$, we use $\lambda_2(G)$ to denote the second largest eigenvalue of the adjacent matrix of $G$.
The following lemma in~\cite{GGS21} illustrates the vertex expansion of regular bipartite graphs by the spectral method.
\begin{lemma}[\cite{GGS21}] \label{lem:vertex-bipartite-expansion}
    For a $\Delta$-regular bipartite graph $G = (V = \calL \sqcup \calR, E)$ such that the second largest eigenvalue of $G$ satisfies that $\lambda_2(G) \le \lambda$, then it holds that for every $S \subseteq \calL$ or $S \subseteq \calR$,
    \begin{align*}
        \abs{N_G(S)} \ge \frac{\Delta^2}{\lambda^2 + (\Delta^2 - \lambda^2) \abs{S}/n} \abs{S}.
    \end{align*}
\end{lemma}

To apply~\Cref{lem:vertex-bipartite-expansion} to $\calG_{n, \Delta}^{\BIP}$, we use the following result by~\cite{BDH22}.
\begin{proposition}[\cite{BDH22}] \label{prop:random-regular-bipartite-spectrum}
    For every fixed $\zeta > 0$, with high probability over $G \sim \calG_{n, \Delta}^{\BIP}$, $\lambda_2(G) \le 2\sqrt{\Delta - 1} + \zeta$.
\end{proposition}

The following corollary is immediate from~\Cref{lem:vertex-bipartite-expansion,prop:random-regular-bipartite-spectrum}.
\begin{corollary} \label{cor:random-regular-bipartite-expansion}
    For every fixed $\zeta > 0$, with high probability over $G = (V = \calL \cup \calR, E) \sim \calG_{n, \Delta}^{\BIP}$, for every $S \subseteq \calL$ or $S \subseteq \calR$,
    \begin{align*}
        \abs{N_G(S)} \ge \frac{\Delta^2}{(2\sqrt{\Delta - 1} + \zeta)^2+ (\Delta^2 - (2\sqrt{\Delta - 1} + \zeta)^2) \abs{S}/n} \abs{S}.
    \end{align*}
\end{corollary}

\subsection{Markov chains}

We state here some basic definitions and lemmas for Markov chains and for more details, please see~\cite{LP17}.

For a state space $\Omega$, a Markov chain $M$ consists of a transition kernel $P$ on $\Omega$ with its stationary distribution $\pi$.
For brevity, we sometimes use the transition matrix $P$ to denote the Markov chain $M$.
In this work, we only consider \emph{ergodic} Markov chains (see details in~\cite{LP17}), and for an ergodic Markov chain $M$, we say it is reversible with respect to $\pi$ if the detailed balanced equation holds:
\begin{align} \label{eq:detailed-balanced-equation}
    \forall \sigma, \tau \in \Omega, \quad \pi(\sigma) P(\sigma, \tau) = \pi(\tau) P(\tau, \sigma).
\end{align}

Recall the Glauber dynamics $P^{\GD}$ with its stationary distribution $\pi$.
It is direct to verify the following fact.
\begin{fact} \label{fact:Glauber-dynamics-reversibility}
    The single-site Glauber dynamics $P^{\GD}$ is reversible w.r.t. its stationary distribution $\pi$.
\end{fact}
\begin{proof}
    We verify~\eqref{eq:detailed-balanced-equation} for $P^{\GD}$ and $\pi$ and it is direct to show the fact.
\end{proof}

We put our attention on the \emph{mixing time} of Markov chains.
For two probability distributions $\mu, \nu$ on $\Omega$, the \emph{total variation distance} is defined as
\begin{align*}
    \TV{\mu}{\nu} \defeq \frac{1}{2} \sum_{x \in \Omega} \abs{\mu(x) - \nu(x)} = \max_{\Lambda \subseteq \Omega} \abs{\mu(\Lambda) - \nu(\Lambda)}.
\end{align*}
For a positive real $\eps \in (0, 1)$, define the \emph{mixing time} of $M$ with error $\eps$ as
\begin{align*}
    \tau_{\mix}(M, \eps) \defeq \inf\set{t > 0 \midc \forall x \in \Omega, \TV{P^t(x, \cdot)}{\pi} \le \eps}.
\end{align*}
We refer to the value of $\tau_{\mix}(M, \cdot)$ when $\eps = 1/4$ as the mixing time of $M$, \IE,
\begin{align*}
    \tau_{\mix}(M) \defeq \tau_{\mix}(M, 1/4).
\end{align*}

For $\sigma, \tau \in \Omega$, we define $Q(\sigma \to \tau) \defeq \pi(\sigma) P(\sigma, \tau)$ and for two subsets $\Lambda, \Pi \subseteq \Omega$, we define
\begin{align*}
    Q(\Lambda \to \Pi) \defeq \sum_{\sigma \in S} \sum_{\tau \in T} Q(\sigma \to \tau).
\end{align*}
The \emph{conductance} of a subset $\Lambda \subseteq \Omega$ is defined as:
\begin{align*}
    \Phi_M(\Lambda) \defeq \frac{Q(\Lambda \to (\Omega \setminus \Lambda))}{\pi(\Lambda)}
\end{align*}
and the conductance of $M$ is defined by
\begin{align*}
    \Phi_M \defeq \inf_{\Lambda \subseteq \Omega : 0 < \pi(\Lambda) \le 1/2} \Phi_M(\Lambda).
\end{align*}
The following lemma shows that the mixing time can be lower bounded by the conductance.
\begin{lemma}[{\cite[Theorem 7.4]{LP17}}] \label{lem:slow-mixing-time}
    For an ergodic Markov chain $M$, it holds that
    \begin{align*}
        \tau_{\mix}(M) \ge \frac{1}{4\Phi_{M}}.
    \end{align*}
\end{lemma}

\subsection{Abstract polymer model} \label{sec:abstract-polymer-model}

In this part, we introduce the \emph{abstract polymer model}.
The abstract polymer model $\calM = (\calC, \sim)$ consists of the collection $\calC$ of \emph{abstract polymers} and a binary relation ${\sim}\subseteq \calC \times \calC$.
A polymer $\gamma = (V_\gamma, w_\gamma)$ is a subset $V_\gamma \subseteq V$ equipped with a weight function $w_\gamma$.
With abuse of notation, we sometimes use $\gamma$ to denote $V_\gamma$ and set the size $\abs{\gamma} \defeq \abs{V_\gamma}$.
For two polymers $\gamma_1, \gamma_2$, we say $\gamma_1$ is \emph{compatible} with $\gamma_2$ if $\gamma_1 \not\sim \gamma_2$ and otherwise they are \emph{incompatible}.
We denote by $\calK$ the collection of finite subsets $\Gamma$ of $\calC$ such that polymers in $\Gamma$ are mutually compatible.
Then we define the partition function of the abstract polymer model $\calM$ as
\begin{align*}
    \Xi(\calM) \defeq \sum_{\Gamma \in \calK} \prod_{\gamma \in \Gamma} w_{\gamma}.
\end{align*}

\paragraph{\bf Example: abstract polymer models induced by a graph}
Mostly we consider the abstract polymer model induced by a graph $G$.
Then we would define a polymer $\gamma = (V_\gamma, w_\gamma)$ via a connected subgraph $V_\gamma$.
We say two polymers $\gamma_1 \sim \gamma_2$ if $V_{\gamma_1} \cup V_{\gamma_2}$ forms a connected subgraph in $G$.

The famous \Kotecky-Preiss condition in~\cite{KP86} is of great significance to compute $\Xi$.
\begin{condition}[\Kotecky-Preiss condition] \label{cond:Kotecky-Preiss-condition}
    We say $(\calC, \sim)$ satisfies \Kotecky-Preiss condition if there exists a function $g : \calC \to \mathbb R_{\ge 0}$ such that for every $\gamma \in \calC$,
    \begin{align*}
        \sum_{\gamma' \sim \gamma} \abs{w_{\gamma'}} \e^{\abs{\gamma'} + g(\gamma')} \le g(\gamma).
    \end{align*}
\end{condition}

We use the following theorem to give an $\FPTAS$ to $\Xi$ together with~\Cref{cond:Kotecky-Preiss-condition}.
\begin{theorem}[Theorem 8 in~\cite{JKP19}] \label{thm:counting-polymer}
    Fix $\Delta > 0$ and a graph family $\calG$ of maximum degree at most $\Delta$.
    Suppose that the followings hold for the polymer models induced by $\calG$ and a function $g(\cdot)$:
    \begin{enumerate}
        \item There exist two constants $C_1, C_2 > 0$ such that for every connected subgraph $\gamma$, it takes at most $\abs{\gamma}^{C_1} \e^{C_2 \abs{\gamma}}$ time to determine whether $\gamma$ is in the model, and then compute $w_\gamma$ and $g(\gamma)$. \label{item:counting-polymer-enumerate}
        \item There exists some constants $C = C(\Delta) > 0$ such that for every $G \in \calG$ and each polymer $\gamma$ induced by $G$, $g(\gamma) \ge C\abs{\gamma}$.\label{item:counting-polymer-function-bound}
        \item \Cref{cond:Kotecky-Preiss-condition} holds with choice $g$.\label{item:counting-polymer-KP-condition}
    \end{enumerate}
    Then there exists an $\FPTAS$ to approximate the partition function $\Xi$ of the abstract polymer model induced by every graph $G \in \calG$ with running time $\bigO{(n/\eps)^{\tau(\Delta)}}$ for $\tau(\Delta) = (\ln{\Delta} + C_2) / C(\Delta)$.
\end{theorem}

\section{Anti-ferromagnetic Potts Models on Random Regular Bipartite Graphs} \label{sec:anti-ferromagnetic-Potts}

In this section, we investigate the geometry of the configuration space of the anti-ferromagnetic $q$-Potts model.
As mentioned above, we aim to show that, at sufficiently low temperature, the anti-ferromagnetic $q$-state Potts models on $\calG_{n, \Delta}^{\BIP}$ follows some ordered phases with high probability.
Briefly, for each configuration in $\Omega$, there exists a coloring pattern such that it assigns almost disjoint colors to each side, and most of vertices agree with the pattern under the configuration.
To describe ordered phases, we formalize the definition of \emph{patterns} as~\cite{JKP19}.
\begin{definition}[Patterns] \label{def:pattern}
    For a color set $[q] = \set{1, \ldots, q}$ and two subsets $A, B \subseteq [q]$, we say $(A, B)$ is a \emph{pattern} if $A$ and $B$ form a partition of $[q]$, \IE, $[q] = A \sqcup B$ and $\abs{A} \neq 0$, $\abs{B} \neq 0$.
    We denote by $\PATTERN$ the collection of all patterns.
\end{definition}

Recall that in~\cite{JKP19}, we say a configuration $\sigma \in \Omega$ \emph{agrees} with $(A, B)$ at $v \in V$ if $\sigma(v) \in A$ when $v \in \calL$ or $\sigma(v) \in B$ when $v \in \calR$; otherwise we say $\sigma$ \emph{disagrees} with it at $v$.
And given a pattern $(A, B)$ and a vertex subset $S \subseteq V$, define $\chi_{A, B}(S)$ as the collection of configurations $\sigma \in \Omega$ such that $\sigma$ agrees with $(A, B)$ at every $v \in V \setminus S$ and disagrees with the pattern at $v \in S$. 
We remark here that unlike the case in~\cite{JKP19}, we allow that the configuration is not necessarily a proper coloring.

To capture the geometric properties of the configuration space $\Omega$, define the following rank function $\RANK : \Omega \to \mathbb N$ by
\begin{align*}
    \forall \sigma \in \Omega, \quad \RANK(\sigma) = \inf\set{\ell \in \mathbb N : \exists (A, B) \in \PATTERN, S \subseteq V, \abs{S} = \ell \land \sigma \in \chi_{A, B}(S)}.
\end{align*}
Furthermore, we separate $\Omega$ according to the ranks $\RANK$ as:
\begin{align*}
    \forall \ell \in \mathbb N, \quad \Omega^{\ell} \defeq \set{\sigma \in \Omega \midc \RANK(\sigma) = \ell}
\end{align*}
and define $\Omega^{\le \ell} \defeq \cup_{t \le \ell} \Omega^{t}, \Omega^{> \ell} \defeq \cup_{t > \ell} \Omega^t$ and the partition functions $Z^{\le \ell} \defeq \sum_{\sigma \in \Omega^{\le \ell}} w(\sigma), Z^{>\ell} \defeq \sum_{\sigma \in \Omega^{>\ell}} w(\sigma)$ restricted on them respectively.
For brevity, given a parameter $\theta \in (0, 1)$, we say a configuration $\sigma$ is of a low rank if $\RANK(\sigma) \le \theta n$ and otherwise it has a high rank.
Moreover, for a configuration $\sigma \in \chi_{A, B}(S)$, we say $\sigma$ is \emph{specified by $(A, B)$ at rank $\abs{S}$}.

Given a configuration $\sigma$, we define $\vecx = \vecx(\sigma) \defeq \set{x_i}_{i = 1}^{q}$ and $\vecy = \vecy(\sigma) \defeq \set{y_i}_{i = 1}^{q}$ as
\[
    x_i = \frac{\abs{\sigma^{-1}(i)\cap \calL}}{n}, \qquad y_i = \frac{\abs{\sigma^{-1}(i) \cap \calR}}{n}.
\]
We also write
\[
    x_S = \sum_{i\in S} x_i, \qquad y_S = \sum_{i\in S}y_i
\]
for convenience.
We have the following lemma for the rank of $\sigma$.

\begin{lemma} \label{lem:rank-identity}
    For every coloring $\sigma$, it holds that
    \begin{equation} \label{eq:rank-overlap}
        \frac{\RANK(\sigma)}{n} = \sum_{i = 1}^q \min\set{x_i, y_i} = 1 - \TV{\vecx}{\vecy}.
    \end{equation}
\end{lemma}
\begin{proof}
Given a pattern $(A, B)\in \PATTERN$, the normalized number of disagreements is $x_B + y_A$. To minimize the number of disagreements, assign color $i$ to $A$ if and only if $y_i \le x_i$, and in the all-tie case split the tied colors between two nonempty palettes. This gives
\[
    \frac{\RANK(\sigma)}n  =\sum_{i=1}^q\min\{x_i,y_i\}\,.
\]
The second equality follows from the property of total variation distance.
\end{proof}

We define the total weight of $\chi_{A, B}(S)$ by
\begin{align*}
    W_{A, B}(S) \defeq \sum_{\sigma \in \chi_{A, B}(S)} w(\sigma).
\end{align*}
A simple but important property of $W_{A, B}(S)$ is the following factorization.
\begin{lemma} \label{lem:2-factorization}
    Fix a pattern $(A, B) \in \PATTERN$.
    For a vertex subset $S \subseteq V$ with its $2$-connected components $S_1, \ldots, S_{\ell}$, it holds that
    \begin{align*}
        \frac{W_{A, B}(S)}{\abs{A}^n \abs{B}^n} = \prod_{i = 1}^{\ell} \frac{W_{A, B}(S_i)}{\abs{A}^n \abs{B}^n}.
    \end{align*}
\end{lemma}
\begin{proof}
    For brevity, let $\Lambda = S \cup N_G(S)$.
    Note that only edges between $S$ and $V \setminus S$ contribute to the Gibbs weights of configurations in $\chi_{A, B}(S)$.
    By direct calculation, it holds that
    \begin{align*}
        \frac{W_{A, B}(S)}{\abs{A}^n \abs{B}^n} &= \frac{1}{\abs{A}^n \abs{B}^n} \sum_{\sigma \in \chi_{A, B}(S)} \e^{-\beta m_{G}(\sigma)} \\
        &= \frac{1}{\abs{A}^n \abs{B}^n} \sum_{\sigma_{\Lambda} \in [q]^{\Lambda} : \exists \sigma \in \chi_{A, B}(S), \sigma(\Lambda) = \sigma_{\Lambda}} \e^{-\beta m_G(\sigma_{\Lambda})} \abs{A}^{n - \abs{\Lambda \cap \calL}} \abs{B}^{n - \abs{\Lambda \cap \calR}} \\
        &= \frac{\sum_{\sigma_{\Lambda} \in [q]^{\Lambda} : \exists \sigma \in \chi_{A, B}(S), \sigma(\Lambda) = \sigma_{\Lambda}} \e^{-\beta m_G(\sigma_{\Lambda})}}{\abs{A}^{\abs{\Lambda \cap \calL}} \abs{B}^{\abs{\Lambda \cap \calR}}}.
    \end{align*}
    Thus we obtain that
    \begin{align*}
        \frac{W_{A, B}(S)}{\abs{A}^n \abs{B}^n} = \prod_{i = 1}^{\ell} \frac{W_{A, B}(S_i)}{\abs{A}^n \abs{B}^n}.
    \end{align*}
    since $S_1, \ldots, S_{\ell}$ are pairwise $2$-disconnected.
\end{proof}

The following lemma gives an upper bound to $W_{A, B}(S)$.
\begin{lemma} \label{lem:pattern-weight-bound}
    For every pattern $(A, B) \in \PATTERN$ and a vertex subset $S \subseteq V$, it holds that
    \begin{align*}
        W_{A, B}(S) \le \abs{A}^n \abs{B}^n \ab(\frac{\abs{A}}{\abs{B}})^{\abs{S \cap \calR} - \abs{S \cap \calL}} \ab(1 - \frac{1 - \e^{-\beta}}{\abs{A}})^{\abs{N_G(S) \cap \calL}} \ab(1 - \frac{1 - \e^{-\beta}}{\abs{B}})^{\abs{N_G(S) \cap \calR}}.
    \end{align*}
\end{lemma}
\begin{proof}
    The proof is similar to~\cite[Lemma 27]{JKP19} with some modifications to adapt the Potts model.
    We calculate $W_{A, B}(S)$ by the following steps using the multiplication rule.
    \begin{itemize}
        \item At first we consider all partial colorings on $S$ (there are at most $\abs{A}^{\abs{S \cap \calR}} \abs{B}^{\abs{S \cap \calL}}$ partial colorings of weight $1$);
        \item secondly we enumerate the vertices $T \subseteq N_G(S)$ with the same colors as their neighbors in $S$ (they contribute $\e^{-\beta\abs{E_G(S, T)}} \le \e^{-\beta\abs{T}}$);
        \item then we consider each vertex $v \in N_G(S) \setminus T$ (there are at most $(\abs{A} - 1)^{\abs{(N_G(S) \cap \calL) \setminus T}} (\abs{B} - 1)^{\abs{(N_G(S) \cap \calR) \setminus T}}$ partial colorings of weight $1$);
        \item lastly we calculate the number of colorings on $V \setminus (S \cup N_G(S))$; there are
        \begin{align*}
            \abs{A}^{n - \abs{S \cap \calL} - \abs{N_G(S) \cap \calL}} \abs{B}^{n - \abs{S \cap \calR} - \abs{N_G(S) \cap \calR}}
        \end{align*}
        partial colorings contributing $1$.
    \end{itemize}
    We combine all steps and arrange it to obtain that
    \begin{align*}
        W_{A, B}(S) &\le \abs{A}^n \abs{B}^n \ab(\frac{\abs{A}}{\abs{B}})^{\abs{S \cap \calR} - \abs{S \cap \calL}} \ab(\frac{\e^{-\beta} + \abs{A} - 1}{\abs{A}})^{\abs{N_{G}(S) \cap \calL}} \ab(\frac{\e^{-\beta} + \abs{B} - 1}{\abs{B}})^{\abs{N_{G}(S) \cap \calR}} \\
        &= \abs{A}^n \abs{B}^n \ab(\frac{\abs{A}}{\abs{B}})^{\abs{S \cap \calR} - \abs{S \cap \calL}} \ab(1 - \frac{1 - \e^{-\beta}}{\abs{A}})^{\abs{N_G(S) \cap \calL}} \ab(1 - \frac{1 - \e^{-\beta}}{\abs{B}})^{\abs{N_G(S) \cap \calR}}.
    \end{align*}
    Thus we conclude the lemma.
\end{proof}

We give the following upper bound for $W_{A, B}(S)$ at low temperature when $\abs{S}$ is bounded.
\begin{lemma} \label{lem:weight-upper-bound}
    Given positive integers $q, \Delta \ge 3$ and a real number $\beta > 0$ satisfying
    \begin{align*}
        \ParameterCondition
    \end{align*}
    and every $\delta > 0$, the following holds with high probability over $G = (V = \calL \cup \calR, E) \sim \calG_{n, \Delta}^{\BIP}$.
    For every pattern $(A, B) \in \PATTERN$ and every vertex subset $S \subseteq V$ of size $\abs{S} \le \theta n$ with $\theta = 1/(6\Delta)$,
    \begin{align*}
        \frac{W_{A, B}(S)}{\abs{A}^n \abs{B}^n} \le \e^{-\kappa \abs{S}}
    \end{align*}
    where
    \begin{align*}
        \kappa = -\ab(\frac{\Delta^2}{4(\Delta - 1) + (\Delta - 2)^2 \theta} - 1)\ln{\ab(1 - \frac{1 - \e^{-\beta}}{q - 1})} - \ln{(q - 1)} - \delta.
    \end{align*}
\end{lemma}
\begin{proof}
    Assume that the statement in~\Cref{cor:random-regular-bipartite-expansion} holds for some real number $\zeta \in (0, 1)$ decided later (this event holds with high probability).

    Fix $\delta > 0$.
    Set two parameters
    \begin{align*}
        K_{\zeta} = \frac{\Delta^2}{(2\sqrt{\Delta - 1} + \zeta)^2 + (\Delta^2 - (2\sqrt{\Delta - 1} + \zeta)^2) \theta} \quad \text{and} \quad K_0 = \frac{\Delta^2}{4(\Delta - 1) + (\Delta - 2)^2 \theta}.
    \end{align*}
    We first verify that $K_0 > 1$.
    In fact, by direct calculation,
    \begin{align*}
        \Delta^2 - \ab(4(\Delta - 1) + (\Delta - 2)^2/(6\Delta)) = \frac{(\Delta - 2)^2(6\Delta - 1)}{6\Delta} > 0
    \end{align*}
    and thus we obtain that $K_0 > 1$.
    Then we can pick a small constant $\eta > 0$ such that
    \[
        \eta < K_0 - 1, \quad -\eta \ln{\ab(1 - \frac{1 - \e^{-\beta}}{q - 1})} < \delta.
    \]
    By continuity, for fixed sufficiently small real number $\zeta > 0$, it holds that $K_{\zeta} \ge K_0 - \eta > 1$.
    For every subset vertex $S \subseteq V$ of size $\abs{S} \le \theta n$, by~\Cref{cor:random-regular-bipartite-expansion},
    \begin{align*}
        \abs{N_G(S) \cap \calL} \ge K_{\zeta} \abs{S \cap \calR} - \abs{S \cap \calL}, \\
        \abs{N_G(S) \cap \calR} \ge K_{\zeta} \abs{S \cap \calL} - \abs{S \cap \calR}.
    \end{align*}
    Then it holds that
    \begin{align*}
        \abs{N_G(S)} = \abs{N_G(S) \cap \calL} + \abs{N_G(S) \cap \calR} \ge (K_{\zeta} - 1) \abs{S}.
    \end{align*}
    By~\Cref{lem:pattern-weight-bound} together with trivial bounds $\abs{A}/\abs{B} \in [1/(q - 1), q - 1], \abs{A} \le q - 1, \abs{B} \le q - 1$ and $\abs{T \cap W} \le \abs{T}$ for every vertex subsets $T, W \subseteq V$,
    \begin{align*}
        \frac{W_{A, B}(S)}{\abs{A}^n \abs{B}^n} &\le (q - 1)^{\abs{S}} \ab(1 - \frac{1 - \e^{-\beta}}{q - 1})^{\abs{N_G(S)}} \\
        &\le \exp\ab(\abs{S}\ab(\ln{(q - 1)} - (K_{\zeta} - 1) \ln \ab(1 - \frac{1 - \e^{-\beta}}{q - 1}))) \\
        &\le \e^{-\kappa \abs{S}}.
    \end{align*}
    Thus we conclude the upper bound.
\end{proof}

The following inequality plays a crucial role in our analysis.
\begin{proposition} \label{prop:KP-parameter}
    Under the assumption $\ParameterCondition$ with choice $\theta = 1/(6\Delta)$, it holds that
    \begin{align} \label{eq:KP-parameter}
        -\ab(\frac{\Delta^2}{4(\Delta - 1) + (\Delta - 2)^2 \theta} - 1) \ln{\ab(1 - \frac{1 - \e^{-\beta}}{q - 1})} - \ln{(q - 1)} > \frac{3}{2} + \ln{\ab(\e \Delta^2 + 2(\Delta^2 + 1))}.
    \end{align}
    Moreover, the gap between two factors is at least $1/4$.
\end{proposition}
\begin{proof}
    Set
    \[
        K_0 = \frac{\Delta^2}{4(\Delta - 1) + (\Delta - 2)^2 \theta}.
    \]
    By $\theta = 1/(6\Delta)$, note that
    \begin{align*}
        K_0 = \frac{\Delta^2}{4(\Delta - 1) + (\Delta - 2)^2 / (6\Delta)} \ge \frac{\Delta^2}{25\Delta/6} = \frac{6\Delta}{25}.
    \end{align*}
    Then with $-\ln{(1 - x)} \ge x$ on $[0, 1]$ and $\ParameterCondition$, it holds that
    \begin{align*}
        -\ab(K_0 - 1) \ln{\ab(1 - \frac{1 - \e^{-\beta}}{q - 1})} - \ln{(q - 1)} &\ge \ab(\frac{6\Delta}{25} - 1) \frac{1 - \e^{-\beta}}{q - 1} - \ln{(q - 1)} \\
        &\ge \frac{54q}{25(q - 1)} \ln{\frac{q\Delta}{1 - \e^{-\beta}}} - \frac{1}{2} - \ln{(q - 1)}.
    \end{align*}
    On the other hand,
    \begin{align*}
        \frac{3}{2} + \ln{\ab(\e \Delta^2 + 2(\Delta^2 + 1))} &= \frac{3}{2} + 2\ln{\Delta} + \ln{(\e + 2 + 2/\Delta^2)} \\
        &\le 2\ln{\Delta} + \frac{3}{2} + \ln{\ab(\e + \frac{17}{8})}.
    \end{align*}
    Then the different between L.H.S. of~\eqref{eq:KP-parameter} and R.H.S. of~\eqref{eq:KP-parameter} is minimized at $\Delta = 4$ and at least
    \begin{align*}
        f(q) = \frac{54q}{25(q - 1)}\ln{q} + \ab(\frac{54q}{25(q - 1)} - 2)\ln{4} - 2 - \ln{(q - 1)} - \ln{\ab(\e + \frac{17}{8})}.
    \end{align*}
    The derivative of $f$ on $[3, +\infty)$ is
    \begin{align*}
        f'(x) = \frac{(54/25 - 1)(x - 1) - (54/25)\ln{(4x)}}{(x - 1)^2}.
    \end{align*}
    The derivative of the numerator of $f'(q)$ is positive on $[3, +\infty)$ and thus the numerator of $f'(x)$ is strictly increasing.
    By direct numerical calculation, we have $f'(7) < 0$ and $f'(8) > 0$, meaning that the minimum of $f(q)$ is $\min\set{f(7), f(8)}$.
    Since $f(7) > 0.2552$ and $f(8) > 0.2593$, we conclude the proposition.
\end{proof}

\subsection{Tail bound at high ranks}

The most important geometric property for anti-ferromagnetic Potts models on random regular bipartite graphs at low temperature is the \emph{tail bound at high ranks}.
\begin{lemma} \label{lem:high-rank-tail-bound}
    Fix positive integers $q \ge 3, \Delta \ge 4$ and a real number $\beta > 0$ satisfying that
    \begin{align*}
        \ParameterCondition.
    \end{align*}
    There exists a constant $C > 0$ such that with high probability over $G \sim \calG_{n, \Delta}^{\BIP}$,
    \begin{align*}
        Z^{> \theta n} \le \e^{-Cn} Z_{G, q}(\beta)
    \end{align*}
    with the choice of $\theta = \frac{1}{6\Delta}$.
\end{lemma}

To show the tail bound, we need the following inequality.
\begin{lemma} \label{lem:balanced-profile-inequality}
    For all probability vectors $\vecx, \vecy \in \Delta_q$ and every $\kappa \ge 16\ceil{q/2}$,
    \begin{align*}
        H(\vecx) + H(\vecy) - \ln{(\floor{q/2} \ceil{q/2})} \le \kappa \inner{\vecx}{\vecy} + 2\e^{-\kappa/(4\ceil{q/2})}.
    \end{align*}
\end{lemma}
\begin{proof}
    Take $\vecz \in \Delta_q$ and $\boldsymbol{\xi} \in \mathbb{R}^q$ as
    \[
        z_i = \frac{x_i + y_i}{2}, \quad \xi_i = \frac{x_i}{x_i + y_i}
    \]
    with convention $\xi_i = 0$ if $x_i + y_i = 0$.
    We construct a pattern $(A, B) \in \PATTERN$ as following: let $J = \set{i \in [q] \;:\; \xi_i \ge 1/2}$.
    If $\abs{J} \le \floor{q/2}$, construct an arbitrary $A$ such that $A \supseteq J$ and $\abs{A} = \floor{q/2}$; otherwise choose $A \subseteq J$ with $\abs{A} = \ceil{q/2}$.
    Let $B = [q] \setminus A$.
    Define two vectors $\vece \in \mathbb{R}^q$ and $\boldsymbol{\zeta} \in \Delta_q$ as
    \[
        e_i = \begin{cases}
            1 - \xi_i & i \in A \\
            \xi_i & i \in B
        \end{cases}\;,
        \quad
        \zeta_i = \begin{cases}
            \frac{1}{2\abs{A}} & i \in A \\
            \frac{1}{2\abs{B}} & i \in B
        \end{cases}\;.
    \]
    Then by $-x_i \ln{x_i} - y_i\ln{y_i} = 2z_i H(\xi_i) - 2z_i\ln{z_i}$, it holds that
    \begin{align*}
        \frac{H(\vecx) + H(\vecy) - \log{\abs{A}\abs{B}}}{2} &= \sum_{i \in A} z_i \ab(H(e_i) + e_i \ln{(\abs{A}/\abs{B})}) \\
        &+ \sum_{i \in B} z_i \ab(H(e_i) - e_i \ln{(\abs{A}/\abs{B})}) + \KL{\vecz}{\boldsymbol{\zeta}}.
    \end{align*}
    Set $r_i = \min\set{e_i, 1 - e_i}$.
    It is direct to see
    \begin{align*}
         \frac{H(\vecx) + H(\vecy) - \log{\abs{A}\abs{B}}}{2} &\le \sum_{i \in [q]} z_i\ab(H(r_i) + r_i \ln{(\ceil{q/2}/\floor{q/2})}) - \KL{\vecz}{\boldsymbol{\zeta}}.
    \end{align*}
    By the binary Gibbs variational inequality,
    \begin{align*}
        H(r) + r \ln{(\floor{q/2}\ceil{q/2})} \le \ln{\ab(1 + \e^{\ln{(\floor{q/2}\ceil{q/2})} - \kappa z})} + \kappa z r.
    \end{align*}
    Then together with $\inner{\vecx}{\vecy} = 4\sum_{i \in [q]} z_i^2 r_i(1 - r_i)$,
    \begin{align*}
         \frac{H(\vecx) + H(\vecy) - \log{\abs{A}\abs{B}}}{2} \le \frac{\kappa}{2} \inner{\vecx}{\vecy} + \sum_{i \in [q]} \ab(z_i \ln{\ab(1 + \frac{\ceil{q/2}}{\floor{q/2}}\e^{-\kappa z_i})} - \zeta_i \varphi(z_i/\zeta_i))
    \end{align*}
    where $\varphi(x) = x\ln{x} - x + 1$.

    Define the residual quantity \[
        R \defeq \sum_{i \in [q]} \ab(z_i \ln{\ab(1 + \frac{\ceil{q/2}}{\floor{q/2}}\e^{-\kappa z_i})} - \zeta_i \varphi(z_i/\zeta_i))
    \]
    Now we show $R \le \e^{-\kappa/(4\ceil{q/2})}$.
    Set $w_i = z_i / \zeta_i$, $\lambda_i = \kappa \zeta_i \ge 8$.
    It holds that
    \[
        R = \sum_{i \in [q]} \zeta_i \ab(w_i \ln{\ab(1 + \frac{\ceil{q/2}}{\floor{q/2}} \e^{-\lambda_i w_i})} - \varphi(w_i))
    \]
    Since $\ceil{q/2}/\floor{q/2} \le 2$ for $q \ge 3$ and $\ln{(1 + x)} \le x$,
    \[
        w \ln{\ab(1 + \frac{\ceil{q/2}}{\floor{q/2}} \e^{-\lambda w})} - \varphi(w) \le 2w\e^{-\lambda w} - \varphi(w).
    \]
    For $\lambda \ge 8$, the function $w \mapsto 2w\e^{-\lambda w} - \varphi(w)$ is negative on $[0, 1/2]$ and $2w\e^{-\lambda w} - \varphi(w) \le 2 (1/2) \e^{-\lambda/2}$ for $w \ge 1/2$.
    Then it holds that
    \begin{align*}
        R \le \sum_{i \in [q]} \zeta_i \e^{-\lambda_i/2} \le \e^{-\kappa/(2\max{\abs{A}, \abs{B}})} = \e^{-\kappa/(4\ceil{q/2})}
    \end{align*}
    by $\lambda_i \ge \kappa/(2\max{\abs{A}, \abs{B}})$ and $\sum_{i \in [q]} \zeta_i = 1$.
    Thus we conclude the lemma.
\end{proof}

Now we are ready to show~\Cref{lem:high-rank-tail-bound}.
\begin{proof}[Proof of~\Cref{lem:high-rank-tail-bound}]
    We work on the configuration model to show that with high probability over $G \sim \calG_{n, \Delta}^{\CM}$,
    \begin{align} \label{eq:configuration-model-tail-bound}
        \Pr[\sigma \sim \mu_{G, q; \beta}]{\RANK(\sigma) > \theta n} \le \e^{-Cn}
    \end{align}
    for some constant $C > 0$ and conclude the lemma by~\Cref{prop:configuration-model-to-random-regular-bipartite-graph}.

    Fix a configuration $\sigma : V \to [q]$ and thus the probability vector $\vecx, \vecy$.
    For convenience, set $N = \Delta n$.
    Let $\vecpi = (\pi_{ij})_{i, j \in [q]}$ be the matrix meaning that there are $n_{ij} = N \pi_{ij}$ edges between color $i$ to color $j$ in $E$ from $\calL$ to $\calR$.
    Then it holds that
    \begin{align*}
        \Pr{\vecpi} = \frac{\prod_{i} (N x_i)! \prod_{j} (N y_j)!}{N! \prod_{i, j} (N \pi_{ij})!}.
    \end{align*}
    Using the inequality for every probability $\vecp \in \Delta_d$,
    \begin{align*}
        (M + 1)^{-1} \e^{M H(\vecp)} \le \binom{M}{Mp_1, \ldots, M p_d} \le \e^{M H(\vecp)},
    \end{align*}
    we obtain that
    \begin{align*}
        \Pr{\vecpi} \le (N + 1)^q \exp\ab(-N(H(\vecpi) - H(\vecx) - H(\vecy))) = (N + 1)^q \e^{-N\KL{\vecpi}{\vecx \otimes \vecy}}.
    \end{align*}
    Consider the data processing inequality under the mapping $(i, j) \mapsto \id{i = j}$.
    We obtain that
    \begin{align*}
        \KL{\vecpi}{\vecx \otimes \vecy} \ge \KL{\Ber(r)}{\Ber(s)}
    \end{align*}
    where $r = \sum_{i} \pi_{ii}$ and $s = \sum_{i} x_i y_i$.
    The Gibbs variatinal formula shows that for every $r \in [0, 1]$ and $\beta > 0$,
    \begin{align*}
        \KL{\vecpi}{\vecx \otimes \vecy} + r \beta &\ge \min_{0 \le x \le 1}\set{\KL{\Ber(x)}{\Ber(s)} + x\ln{\beta}} \\
        &= -\ln{\ab(1 - s + s \e^{-\beta})} \\
        &\le (1 - \e^{-\beta}) s.
    \end{align*}
    At $\beta = \infty$, it must hold that $r = 0$ (otherwise $\sigma$ is not a proper coloring).
    Then it holds that
    \begin{align*}
        \KL{\vecpi}{\vecx \otimes \vecy} \ge \KL{\Ber(0)}{\Ber(s)} = -\ln{(1 - s)} \ge (1 - \e^{-\beta}) s.
    \end{align*}
    We mark here that when $s = 1$, the argument still holds since $\sigma$ is colored by only one color.

    Now we compute the expectation of the weight $\sigma$.
    Note that there are at most $(N + 1)^{(q - 1)^2}$ possible feasible $\vecpi$ and the weight of $\sigma$ corresponding to $\vecpi$ is $\e^{-\beta N r}$.
    Then,
    \begin{align*}
        \E[G \sim \calG_{n, \Delta}^{\CM}]{\e^{-\beta m_G(\sigma)}} &\le (N + 1)^{q + (q - 1)^2} \e^{-N(1 - \e^{-\beta}) \inner{\vecx}{\vecy}} \\
        &\le (\Delta n + 1)^{q^2} \exp\ab(-\Delta n \ab(1 - \e^{-\beta})\inner{\vecx}{\vecy}).
    \end{align*}

    Note that there are at most $(n + 1)^{2(q - 1)}$ possible $\vecx, \vecy$ and each pair of $\vecx, \vecy$ can correspond to at most $\e^{n(H(\vecx) + H(\vecy))}$ configurations.
    It holds that
    \begin{align*}
        \E[\calG_{n, \Delta}^{\CM}]{Z^{\ge \theta n}} &\le (n + 1)^{2(q - 1)} (\Delta n + 1)^{q^2} \max_{\RANK > \theta n} \set{\exp\ab(n(H(\vecx) + H(\vecy) - \Delta (1 - \e^{-\beta})\inner{\vecx}{\vecy}))}.
    \end{align*}
    Apply~\Cref{lem:balanced-profile-inequality} for $\kappa = \frac{3\Delta(1 - \e^{-\beta})}{4} \ge 16\ceil{q/2}$ by assumption, and we obtain that
    \begin{align*}
        H(\vecx) + H(\vecy) - \Delta(1 - \e^{-\beta}) \inner{\vecx}{\vecy} &\le \ln{(\floor{q/2} \ceil{q/2})} - \frac{\Delta(1 - \e^{-\beta})}{4}\inner{\vecx}{\vecy} + 2\exp\ab(-\frac{\Delta(1 - \e^{-\beta})}{16\ceil{q/2}}) \\
        &\le \ln{(\floor{q/2} \ceil{q/2})} - \frac{\Delta(1 - \e^{-\beta})}{4} \frac{\theta^2}{q} + 2\exp\ab(-\frac{\Delta(1 - \e^{-\beta})}{16\ceil{q/2}}).
    \end{align*}
    Set
    \[
        C \defeq \frac{1 - \e^{-\beta}}{144q \Delta} - 2\exp\ab(-\frac{\Delta(1 - \e^{-\beta})}{16\ceil{q/2}}).
    \]
    Then it holds that
    \[
        \E[\calG_{n, \Delta}^{\CM}]{Z^{\ge \theta n}} \le (n + 1)^{2(q - 1)} (\Delta n + 1)^{q^2} (\floor{q/2}\ceil{q/2})^{n} \e^{-C n}
    \]
    and by Markov inequality and a crude estimation to $Z_{G, q}(\beta) \ge (\floor{q/2}\ceil{q/2})^n$, we prove~\eqref{eq:configuration-model-tail-bound} and thus conclude the lemma if $C > 0$.
    In fact, by $\ParameterCondition$ and direct numerical calculation, it holds that
    \begin{align*}
        2\exp\ab(-\frac{\Delta(1 - \e^{-\beta})}{16\ceil{q/2}}) \le \exp\ab(-\frac{81}{32} \ln{\frac{q\Delta}{1 - \e^{-\beta}}}) \le \exp\ab(-\ln{\frac{288 q\Delta}{1 - \e^{-\beta}}}).
    \end{align*}
    Therefore we can obtain that $C > 0$.
\end{proof}

\section{Slow Mixing of Glauber Dynamics} \label{sec:slow-mixing-theorem}

In this section, we show how to prove the torpid mixing of the single-site Glauber dynamics $P^{\GD}$ for anti-ferromagnetic $q$-Potts models on random regular bipartite graphs $\calG_{n, \Delta}^{\BIP}$ (\IE, \Cref{thm:Glauber-dynamics-slow-mixing}).
The key ingredient is to pick an appropriate subset of $\Omega$ and to show that is has an exponentially small conductance.

Fix $\theta = 1/(6\Delta)$.
For a pattern $(A, B) \in \PATTERN$, set $\Omega_{A, B}^{\LOW}$ collecting all configurations specified by $(A, B)$ at low ranks:
\[
    \Omega_{A, B}^{\LOW} \defeq \bigcup_{S \subseteq V : \abs{S} \le \theta n} \chi_{A, B}(S).
\]
Moreover, define the inner-boundary of $\Omega_{A, B}^{\LOW}$ as
\begin{align} \label{eq:inner-boundary}
    \calB_{A, B} = \set{\sigma \in \Omega_{A, B}^{\LOW} \mid \exists S \subseteq V, \abs{S} = \theta n, \sigma \in \chi_{A, B}(S)}.
\end{align}
We first show how to pick a proper pattern $(A, B)$ and then derive a small conductance of $\Omega_{A, B}^{\LOW}$.
\begin{proposition} \label{prop:pattern-choice}
    Assume that the statement in~\Cref{lem:high-rank-tail-bound} holds.
    There exists a pattern $(A, B) \in \PATTERN$ such that for all sufficiently large positive integer $n$,
    \[
        \frac{1}{2^{q + 1} - 4} \le \mu\ab(\Omega_{A, B}^{\LOW}) \le \frac{1}{2}.
    \]
\end{proposition}
\begin{proof}
    For every pattern $(A, B) \in \PATTERN$, the upper bound comes directly from the symmetry of colors and $\theta = 1/(6\Delta) \le 1/24$.
    To show the lower bound, by~\Cref{lem:high-rank-tail-bound}, it holds that
    \[
        \mu\ab(\bigcup_{(A, B) \in \PATTERN} \Omega_{A, B}^{\LOW}) \ge 1 - \e^{-Cn}
    \]
    for some constant $C > 0$.
    Pick the pattern $(A, B) \in \PATTERN$ maximizing $\mu(\Omega_{A, B}^{\LOW})$.
    By a union bound,
    \begin{align*}
        \mu\ab(\Omega_{A, B}^{\LOW}) &\ge \frac{1}{\abs{\PATTERN}} \mu\ab(\bigcup_{(A, B) \in \PATTERN} \Omega_{A, B}^{\LOW}) \\
        &\ge \frac{1 - \e^{-Cn}}{2^{q} - 2} \\
        &\ge \frac{1}{2^{q + 1} - 4}.
    \end{align*}
    Thus we conclude the proposition.
\end{proof}

The following lemma shows that when the vertex expansion holds, the inner boundary $\calB_{A, B}$ has a small Gibbs weight.
\begin{lemma} \label{lem:small-inner-boundary}
    Assume that the statement in~\Cref{lem:weight-upper-bound} holds.
    There exists a constant $C > 0$ such that
    \begin{align*}
        \mu\ab(\calB_{A, B}) \le \e^{-C n}
    \end{align*}
    for all sufficiently large positive integer $n$.
\end{lemma}
\begin{proof}
    By~\Cref{lem:weight-upper-bound} with a small $\delta > 0$ and~\Cref{prop:KP-parameter}, for every $S \subseteq V$ of size $\theta n$, it holds that
    \begin{align*}
        \frac{W_{A, B}(S)}{\abs{A}^n \abs{B}^n} \le \e^{-\rho\abs{S}}
    \end{align*}
    where $\rho > \frac{3}{2} + \ln{(\e\Delta^2 + 2(\Delta^2 + 1))}$.
    Then it holds that
    \begin{align*}
        \sum_{\sigma \in \calB_{A, B}} w(\sigma) &= \abs{A}^n \abs{B}^n \sum_{S \subseteq V : \abs{S} = \theta n} \frac{W_{A, B}(S)}{\abs{A}^n \abs{B}^n} \\
        &\le \abs{A}^n \abs{B}^n \binom{2n}{\theta n} \e^{-\rho \theta n} \\
        &\le \abs{A}^n \abs{B}^n \ab(\frac{2\e n}{\theta n})^{\theta n} \e^{-\rho \theta n} \\
        &\le \abs{A}^n \abs{B}^n \e^{-C n} \\
        &\le \e^{-Cn} Z_{G, q}(\beta)
    \end{align*}
    where $C \ge \frac{\theta}{4}(\rho - \ln{(2\e/\theta)}) > \frac{\theta}{4}\ab(\frac{3}{2} + \ln{(\e\Delta^2 + 2(\Delta^2 + 1))} - \ln{(2\e/\theta)}) > 0$.
\end{proof}

The choice of $(A, B)$ in~\Cref{prop:pattern-choice} together with~\Cref{lem:small-inner-boundary} implies a small conductance.
\begin{lemma} \label{lem:small-conductance}
    Assume that statements in~\Cref{lem:high-rank-tail-bound,lem:weight-upper-bound} hold.
    Let $(A, B)$ be the pattern in~\Cref{prop:pattern-choice}.
    Then there exists a constant $C > 0$ such that
    \[
        \Phi_{P^{\GD}}\ab(\Omega_{A, B}^{\LOW}) \le (2^{q + 1} - 4)\e^{-Cn}
    \]
    for all sufficiently large positive integer $n$.
\end{lemma}
\begin{proof}
    For a transition $\sigma \to \tau$ where $\sigma \in \Omega_{A, B}^{\LOW}$ and $\tau \notin \Omega_{A, B}^{\LOW}$ by $P^{\GD}$, since the transition changes only one coordinate, it must hold that $\sigma \in \calB_{A, B}$.
    By definition of $Q$, it holds that
    \begin{align*}
        Q\ab(\Omega_{A, B}^{\LOW} \to \Omega \setminus \Omega_{A, B}^{\LOW}) &= \sum_{\sigma \in \calB_{A, B}} \mu(\sigma) P^{\GD}\ab(\sigma \to \Omega \setminus \Omega_{A, B}^{\LOW}) \le \mu\ab(\calB_{A, B}).
    \end{align*}
    Thus we have
    \begin{align*}
        \Phi_{P^{\GD}}\ab(\Omega_{A, B}^{\LOW}) &\le \frac{\mu\ab(\calB_{A, B})}{\mu\ab(\Omega_{A, B}^{\LOW})}.
    \end{align*}
    By the lower bound in~\Cref{prop:pattern-choice} and the upper bound for $\mu(\calB_{A, B})$ in~\Cref{lem:small-inner-boundary}, we conclude the lemma.
\end{proof}

\begin{proof}[Proof of~\Cref{thm:Glauber-dynamics-slow-mixing}]
    Assume that statements in~\Cref{lem:high-rank-tail-bound,lem:weight-upper-bound} hold (the event occurs with high probability).
    Pick the pattern $(A, B) \in \PATTERN$ in~\Cref{prop:pattern-choice} and set $\Lambda = \Omega_{A, B}^{\LOW}$.
    By~\Cref{prop:pattern-choice,lem:small-conductance}, it holds that $\mu(\Lambda) \le 1/2$ and $\Phi_{P^{\GD}}(\Lambda) \le \e^{-Cn}$ with some factor $C = C(q, \Delta, \beta)$ for all sufficiently large integer $n$.
    Hence $\Phi_{P^{\GD}} \le \Phi_{P^{\GD}}(\Lambda) \le \e^{-Cn}$ and by~\Cref{lem:slow-mixing-time}, \Cref{thm:Glauber-dynamics-slow-mixing} is concluded.
\end{proof}

\section{Approximately Counting Anti-ferromagnetic Potts Models} \label{sec:counting-theorem}

In this section, we design an $\FPTAS$ for counting the partition function of anti-ferromagnetic Potts models on random bipartite regular graphs (\Cref{thm:counting-results}).
To give an approximation to $Z_{G, q}(\beta)$, we apply the following ideas which can be viewed as a generalization of the algorithm to count the number proper colorings in~\cite{JKP19}.
\begin{enumerate}
    \item first, by~\Cref{lem:high-rank-tail-bound}, with high probability the major weight of $Z_{G, q}(\beta)$ lies in configurations of rank at most $\theta n$ for $\theta = 1/(6\Delta)$;
    \item then, the quantity $Z^{\le \theta n}$ can be approximately estimated by summing up configurations specified by each pattern respectively;
    \item finally, for a fixed pattern, the method of the abstract polymer model provides an algorithm to estimate the weights of specified configurations.
\end{enumerate}

To formalize the proof idea, set $\theta = 1/(6\Delta)$.
Define $\calS = \calS(G, \theta)$ as the collection of vertex subsets whose $2$-connected components are of size at most $\theta n$, \IE,
\begin{align*}
    \calS \defeq \set{S \subseteq V \mid \text{each $2$-connected component in $G$ is of size at most $\theta n$}}.
\end{align*}
For each pattern $(A, B) \in \PATTERN$, define the quantity $Z_{A, B}(\calS)$ as
\begin{align*}
    Z_{A, B}(\calS) \defeq \sum_{S \in \calS} W_{A, B}(S).
\end{align*}
Recall the definition of $\Omega_{A, B}^{\LOW}$ for $(A, B)$.
For convenience, we define the quantity
\begin{align*}
    Z_{A, B}^{\LOW} \defeq \sum_{\sigma \in \Omega_{A, B}^{\LOW}} w(\sigma) = \sum_{S \subseteq V : \abs{S} \le \theta n} W_{A, B}(S).
\end{align*}

The following lemmas together with~\Cref{lem:high-rank-tail-bound} are key ingredient to show~\Cref{thm:counting-results}.
\begin{lemma} \label{lem:estimation-by-patterns}
    Given positive integers $q \ge 3, \Delta \ge 4$ and a real number $\beta > 0$ satisfying
    \[
        \ParameterCondition,
    \]
    there exists a constant $C > 1/(20q)$ such that
    \begin{align*}
        1 \le \frac{\sum_{(A, B) \in \PATTERN} Z_{A, B}^{\LOW}}{Z^{\le \theta}} \le 1 + \binom{2^q - 2}{2} (2n + 1) \e^{-Cn}.
    \end{align*}
\end{lemma}

\begin{lemma} \label{lem:pattern-approximations}
    Assume that statements in~\Cref{lem:weight-upper-bound,prop:KP-parameter} hold for positive integers $q \ge 3, \Delta \ge 4$ and a real number $\beta > 0$ satisfying
    \[
        \ParameterCondition.
    \]
    There exists a constant $C > 0$ such that for all sufficiently large positive integer $n$ and every pattern $(A, B) \in \PATTERN$,
    \begin{align*}
        1 \le \frac{Z_{A, B}(\calS)}{Z_{A, B}^{\LOW}} \le 1 + \e^{-Cn}.
    \end{align*}
\end{lemma}

\begin{lemma} \label{lem:polymer-count-sparse}
    Assume that statements in~\Cref{lem:weight-upper-bound,prop:KP-parameter} hold for positive integers $q \ge 3, \Delta \ge 4$ and a real number $\beta > 0$ satisfying $\ParameterCondition$.
    For every real number $\eps \in (0, 1)$ and every pattern $(A, B) \in \PATTERN$, there exists a deterministic algorithm such that it outputs an $\eps$-approximation to $Z_{A, B}(\calS)$ in time $(n/\eps)^{\tau(\Delta)}$ for a computable function $\tau : \mathbb{R} \to (0, \infty)$ depending only on $\Delta$.
\end{lemma}

We provide the proof of~\Cref{lem:estimation-by-patterns} in~\Cref{sec:estimation-by-patterns}, the proof of~\Cref{lem:pattern-approximations} in~\Cref{sec:pattern-approximations} and the proof of~\Cref{lem:polymer-count-sparse} in~\Cref{sec:polymer-count-sparse}.
Here we show~\Cref{thm:counting-results}.
\begin{proof}[Proof of~\Cref{thm:counting-results}]
    We assume that all statements in~\Cref{lem:high-rank-tail-bound,lem:weight-upper-bound,prop:KP-parameter} hold (this event occurs with high probability).
    Set $\theta = 1/(6\Delta)$.
    For all sufficiently large positive integer $n$, by~\Cref{lem:high-rank-tail-bound,lem:estimation-by-patterns,lem:pattern-approximations}, there exists a constant $C = C(q, \Delta, \beta) > 0$ such that the quantity
    \begin{align} \label{eq:final-approximation}
        1 - \e^{-Cn} \le \frac{\sum_{(A, B) \in \PATTERN} Z_{A, B}(\calS)}{Z_{G, q}(\beta)} \le 1 + 4\e^{-Cn}.
    \end{align}
    
    When $\eps \le 20\e^{-Cn}$, note that $q^{2n} = \poly{n, \eps^{-1}}$ and it takes $\poly{n}$ time to compute the weight of each configuration.
    Therefore, we enumerate all configurations by brute force and directly calculate $Z_{G, q}(\beta)$.
    Suppose that $\eps > 20\e^{-Cn}$.
    We apply the algorithm in~\Cref{lem:polymer-count-sparse} to output an $(\eps/5)$-approximation $\wh{Z}_{A, B}$ to $Z_{A, B}(\calS)$ in time $(n/\eps)^{\tau(\Delta)}$ for each pattern $(A, B) \in \PATTERN$ and write $\wh{Z} = \sum_{(A, B) \in \PATTERN} \wh{Z}_{A, B}$.
    By~\eqref{eq:final-approximation}, it holds that
    \begin{align*}
        1 - \eps \le (1 - \eps/5)(1 - \e^{-Cn}) \le \frac{\wh{Z}}{Z_{G, q}(\beta)} \le (1 + 4\e^{-Cn})(1 + \eps/5) \le 1 + \eps,
    \end{align*}
    meaning that $\wh{Z}$ is an $\eps$-approximation to $Z_{G, q}(\beta)$.
    By~\Cref{lem:polymer-count-sparse}, the total running time is at most $(2^q - 2) (n/\eps)^{\tau(\Delta)}$.
    Then the theorem is concluded.
\end{proof}

The proof of~\Cref{cor:proper-coloring-result} comes directly from~\Cref{thm:counting-results}.
\begin{proof}[Proof of~\Cref{cor:proper-coloring-result}]
    Note that the whole analysis holds for proper vertex-coloring models whenever $\Delta \ge 9q\ln{(q\Delta)}$.
    Then we can obtain the $\FPTAS$ similarly.
\end{proof}

\subsection{Estimation by patterns} \label{sec:estimation-by-patterns}

The proof of~\Cref{lem:estimation-by-patterns} is implicit stated in~\cite{JKP19,GGS21}.
For completeness, we formally state here.

The following proposition bounds the intersection between two different patterns.
\begin{proposition} \label{prop:pattern-intersection}
    For two distinct patterns $(A, B), (C, D) \in \PATTERN$, there exists a constant $C > 1/(20q)$ such that
    \begin{align*}
        \sum_{\sigma \in \Omega_{A, B}^{\LOW} \cap \Omega_{C, D}^{\LOW}} w(\sigma) \le (2n + 1) \e^{-Cn} Z^{\le \theta n}.
    \end{align*}
\end{proposition}
\begin{proof}
    For a configuration $\sigma \in \Omega_{A, B}^{\LOW} \cap \Omega_{C, D}^{\LOW}$, the union of its disagree part in at least one side is of size at most $2\theta n$.
    Outside of the disagree part, colors on $\calL$ must lie in $A \cap C$ and colors on $\calR$ must lie in $B \cap D$.
    Note that $\abs{A \cap C} + \abs{B \cap D} \le q - 1$ (otherwise they are same).
    Then counting $\sigma$ by firstly enumerating the disagree part and secondly colors, together with $w(\sigma) \le 1$, it holds that
    \begin{align*}
        \sum_{\sigma \in \Omega_{A, B}^{\LOW} \cap \Omega_{C, D}^{\LOW}} w(\sigma) \le (\floor{(q - 1)/2}\ceil{(q - 1)/2})^{n}(2n + 1) \e^{2nH(\theta) + 2\theta n \ln{q}}.
    \end{align*}
    Therefore, pick $C = \ln{\frac{\floor{q/2}\ceil{q/2}}{\floor{(q - 1)/2}\ceil{(q - 1)/2}}} - 2H(\theta) - 2\theta \ln{q}$ and we conclude the proposition if $C > 1/(20q)$.
    In fact, it holds that
    \begin{align*}
        \ln{\frac{\floor{q/2}\ceil{q/2}}{\floor{(q - 1)/2}\ceil{(q - 1)/2}}} \ge 2/q
    \end{align*}
    and by $\ParameterCondition$,
    \begin{align*}
        2H(\theta) + 2\theta \ln{q} \le \frac{\ln{(6\e q\Delta)}}{3\Delta} \le \frac{1}{18q}.
    \end{align*}
    Therefore we obtain that $C > 1/(20q)$.
\end{proof}

\begin{proof}[Proof of~\Cref{lem:estimation-by-patterns}]
    The lower bound is direct: it is trivial to see $Z^{\le \theta n} \le \sum_{(A, B) \in \PATTERN} Z_{A, B}^{\LOW}$.
    For the upper bound, by~\Cref{prop:pattern-intersection},
    \begin{align*}
        \sum_{(A, B) \in \PATTERN} Z_{A, B}^{\LOW} - Z^{\le \theta n} &= \sum_{(A, B) \neq (C, D)} \sum_{\sigma \in \Omega_{A, B}^{\LOW} \cap \Omega_{C, D}^{\LOW}} w(\sigma) \\
        &\le \binom{2^q - 2}{2} (2n + 1) \e^{-Cn} Z^{\le \theta n}.
    \end{align*}
    Thus we conclude the lemma.
\end{proof}

\subsection{Pattern approximations} \label{sec:pattern-approximations}

To show~\Cref{lem:pattern-approximations}, the following corollary shows that the upper bound in~\Cref{lem:weight-upper-bound} also holds for every subset $S \in \calS$.
\begin{corollary} \label{cor:spread-weight-upper-bound}
    Assume that the statement in~\Cref{lem:weight-upper-bound} holds.
    Then it holds that for every pattern $(A, B) \in \PATTERN$ and every vertex subset $S \in \calS$, it holds that
    \begin{align*}
        \frac{W_{A, B}(S)}{\abs{A}^n \abs{B}^n} \le \e^{-\kappa \abs{S}}
    \end{align*}
    where $\kappa$ is the exact factor in~\Cref{lem:weight-upper-bound}.
\end{corollary}
\begin{proof}
    For $S \in \calS$, let its maximal $2$-connected components be $S_1, \ldots, S_{\ell}$.
    By assumption $\abs{S_i} \le \theta n$ for $i = 1, \ldots, \ell$.
    We bound $W_{A, B}(S)$ by~\Cref{lem:2-factorization} and~\Cref{lem:weight-upper-bound} to show that
    \begin{align*}
        \frac{W_{A, B}(S)}{\abs{A}^n \abs{B}^n} &\le \prod_{i = 1}^{\ell} \frac{W_{A, B}(S_i)}{\abs{A}^n \abs{B}^n} \\
        &\le \prod_{i = 1}^{\ell} \e^{-\kappa \abs{S_i}} \\
        &= \e^{-\kappa \abs{S}}.
    \end{align*}
\end{proof}

\begin{proof}[Proof of~\Cref{lem:pattern-approximations}]
    Recall the definition of $\Omega_{A, B}^{\LOW}$.
    It is trivial to see $Z_{A, B}^{\LOW} \le Z_{A, B}(\calS)$.
    For the upper bound, pick a sufficiently small $\delta > 0$.
    Apply~\Cref{cor:spread-weight-upper-bound} and we obtain that
    \begin{align*}
        Z_{A, B}(\calS) - Z_{A, B}^{\LOW} &\le \abs{A}^{n} \abs{B}^{n} \sum_{S \ge \theta n : S \in \calS} \e^{-\kappa \abs{S}} \\
        &\le \sum_{k = \theta n + 1}^{2n} \binom{2n}{k} \e^{-\kappa \cdot k}.
    \end{align*}
    Since~\Cref{prop:KP-parameter} holds, it can be shown that $\kappa > \ln{(12\e\Delta)} > \ln{((2 - \theta)/\theta)}$.
    Therefore, the sequence $\set{\binom{2n}{k} \e^{-\kappa \cdot k}}_{k \ge \theta n + 1}$ decreases geometrically and by $\binom{2n}{\theta n} \le \exp(2nH(\theta/2))$, it holds that
    \begin{align*}
        Z_{A, B}(\calS) - Z_{A, B}^{\LOW} &\le 2n \abs{A}^{n} \abs{B}^{n} \exp(2nH(\theta/2) - \theta \kappa n) \\
        &\le 2n \exp(2nH(\theta/2) - \theta \kappa n) Z_{A, B}^{\LOW}.
    \end{align*}
    By a direct calculation, $2H(\theta/2) \le \theta \ln{(2\e/\theta)} \le \theta\ln{(2\e\Delta)} < \theta \cdot \kappa$.
    Therefore, there exists a factor $C > 0$ such that $Z_{A, B}(\calS) \le (1 + \e^{-Cn}) Z_{A, B}^{\LOW}$.
\end{proof}

\subsection{Counting by abstract polymers} \label{sec:polymer-count-sparse}

In this part, we provide our algorithm in~\Cref{lem:polymer-count-sparse} to approximately compute $Z_{A, B}(\calS)$ for each pattern $(A, B) \in \PATTERN$.
The main tool is the \emph{abstract polymer model} introduced in~\Cref{sec:abstract-polymer-model}.
Fix $\theta = 1/(6\Delta)$ and a pattern $(A, B) \in \PATTERN$.
We construct the abstract polymer model $\calM_{A, B}(\theta) = (\calC = \calC_{A, B}(\theta), \sim)$ as following.
For every connected subgraph $S \subseteq V$ in $G^2$ of size $\abs{S} \le \theta n$, it induces a polymer $\gamma = (V_\gamma, w_\gamma)$ with $V_{\gamma} = S$ and
\begin{align*}
    w_{\gamma} \defeq \frac{W_{A, B}(V_\gamma)}{\abs{A}^n \abs{B}^n}.
\end{align*}
For two polymers $\gamma_1, \gamma_2$, we say $\gamma_1 \sim \gamma_2$ if the induced subgraph $G^2[\gamma_1 \cup \gamma_2]$ is connected.

For the model of proper colorings,~\cite{JKP19} shows that the partition function of the polymer model is somehow a `normalized' number of proper colorings at low ranks.
For the anti-ferromagnetic Potts model, we put the following proposition showing the relation between $\Xi(\calM_{A, B}(\theta))$ and $Z_{A, B}(\calS)$ stated in~\cite{GGS21} and for completeness we provide the proof.
\begin{proposition} \label{prop:polymer-partition-function}
    For every pattern $(A, B) \in \PATTERN$, it holds that
    \begin{align*}
        Z_{A, B}(\calS) = \abs{A}^n \abs{B}^n \cdot \Xi\ab(\calM_{A, B}(\theta)).
    \end{align*}
\end{proposition}
\begin{proof}
    Observe that there is a one-to-one correspondence between $\calS$ and compatible polymers.
    Then the proposition holds directly from~\Cref{lem:2-factorization}.
\end{proof}

\subsubsection{Verify \Kotecky-Preiss condition}

Now we verify the \Kotecky-Preiss condition (\Cref{cond:Kotecky-Preiss-condition}) and thus we obtain an $\FPTAS$ for computing $\Xi(\calM_{A, B}^{\theta})$ by~\Cref{thm:counting-polymer}.
\begin{lemma} \label{lem:KP-threshold}
    Let $\Xi = (\calC, \sim)$ be an abstract polymer induced by a graph of degree at most $d$.
    Suppose that
    \begin{align*}
        \forall \gamma \in \calC, \quad w_{\gamma} \le \e^{-\rho \abs{\gamma}}
    \end{align*}
    where $\rho \ge \frac{3}{2} + \ln{(\e \cdot d + 2(d + 1))}$.
    Then~\Cref{cond:Kotecky-Preiss-condition} holds for $g(\gamma) = \abs{\gamma}/2$.
\end{lemma}
\begin{proof}
    We use the typical propostion that given a graph $G = (V, E)$ of maximum degree $d$, for every vertex $v \in V$ and a positive integer $k \ge 1$, the number of connect subsets containing $v$ of size $k$ is at most $(\e d)^{k - 1}$.
    Then it holds that
    \begin{align*}
        \sum_{\gamma' \sim \gamma} \abs{w_{\gamma'}} \e^{\abs{\gamma'} + g(\gamma')} &\le (d + 1) \abs{\gamma} \sum_{k \ge 1} (\e d)^{k - 1} \e^{-(\rho - 3/2)k} \\
        &= \frac{(d + 1) \e^{-(\rho - 3/2)}}{1 - \e d \cdot \e^{-(\rho - 3/2)}} \abs{\gamma}.
    \end{align*}
    By assumption, $\e^{-(\rho - 3/2)} \le \frac{1}{\e d + 2(d + 1)}$ and thus
    \[
         \sum_{\gamma' \sim \gamma} \abs{w_{\gamma'}} \e^{f(\gamma') + g(\gamma')} = \frac{(d + 1) \e^{-(\rho - 3/2)}}{1 - \e d \cdot \e^{-(\rho - 3/2)}} \abs{\gamma} \le \frac{\abs{\gamma}}{2}.
    \]
    The~\Cref{cond:Kotecky-Preiss-condition} holds.
\end{proof}

\begin{proof}[Proof of~\Cref{lem:polymer-count-sparse}]
    Fix a pattern $(A, B) \in \PATTERN$ and we construct the abstract polymer model $\calM_{A, B}(\theta) = (\calC = \calC(\theta), \sim)$ as before.
    By~\Cref{lem:weight-upper-bound} and~\Cref{prop:KP-parameter},~\Cref{cond:Kotecky-Preiss-condition} holds for $g(\gamma) = \abs{\gamma}/2$ by~\Cref{lem:KP-threshold} with $d = \Delta^2$.
    Obviously \Cref{item:counting-polymer-enumerate,item:counting-polymer-function-bound} hold by definition and our construction.
    Then by~\Cref{thm:counting-polymer}, there exists a deterministic algorithm to output an $\eps$-approximation to $\Xi(\calM_{A, B})$ running in time $(n/\eps)^{\tau(\Delta)}$ for a computable function $\tau : \mathbb R \to (0, +\infty)$ depending only on $\Delta$.
    Applying this algorithm to estimate $\Xi(\calM_{A, B})$ and by~\Cref{prop:polymer-partition-function}, we obtain the desired algorithm.
\end{proof}

\section{Conclusion} \label{sec:conclusion}

In this work, we study the anti-ferromagnetic $q$-Potts models on random regular bipartite graphs $\calG_{n, \Delta}^{\BIP}$.
Though the celebrated single-site Glauber dynamics exhibits a slow mixing time, there still might be an $\FPTAS$ for approximating the partition function $Z_{G, q}(\beta)$ in the non-uniqueness regime.
Both negative and positive results rely on the geometric properties of the model.
We state some discussion and conjectures here.

\subsection{Geometric properties of anti-ferromagnetic Potts models}

In~\Cref{sec:anti-ferromagnetic-Potts}, we demonstrate the configuration space according to ranks when $\beta$ is in the non-uniqueness regime.
The key observation is that, for anti-ferromagnetic Potts models on random regular bipartite graphs, when it is in the non-uniqueness threshold, most of configurations are concentrated around some ground states characterized by different patterns.
It has been mentioned above that such properties are inspired by the existence of long-range order of anti-ferromagnetic $3$-state Potts models on lattices $\mathbb Z^d$ conjectured by~\cite{Kotecky85} and answered by~\cite{FS19} for sufficiently large $d$.
However, it still remains an understanding when such phase transition occurs for anti-ferromagnetic $q$-Potts models on random regular bipartite graphs and we ask the following question.
\begin{question} \label{que:phase-transition}
    For anti-ferromagnetic $q$-state Potts models on random $\Delta$-regular bipartite graphs, when does the model undergo a phase transition?
\end{question}
We remark here that \Cref{que:phase-transition} is important since our results in this work actually rely on the property that the model is somehow in an ordered phase.

\subsection{Torpid and rapid mixing of Markov chains}

\Cref{thm:Glauber-dynamics-slow-mixing} shows the torpid mixing of Glauber dynamics for anti-ferromagnetic $q$-Potts models when $\beta$ is in the non-uniqueness threshold.
It has been shown in~\cite{GKRS15} that the local-update Markov chain (including the Glauber dynamics) exhibits torpid mixing for $3$-colorings on lattices by the existence of phase transition. 
We conjecture that the coexistence of different phases implies the torpid mixing of this family of Markov chains.
\begin{conjecture}
    For anti-ferromagnetic $q$-Potts models where the phase transition occurs, the local-update Markov chains (including the single-site Glauber dynamics $P^{\GD}$) exhibit exponential mixing time at sufficiently low temperature.
\end{conjecture}

On the other side, in spite of the torpid mixing of $P^{\GD}$, it remains open whether there exists a fast Markov chain to design an $\FPAUS$ for $\mu_{G, q; \beta}$.
\begin{conjecture}
    For anti-ferromagnetic $q$-Potts models on random $\Delta$-regular $2n$-vertex bipartite graphs at $\beta > 0$, with high probability there exists an $\FPAUS$ for $\mu_{G, q; \beta}$ running in time $\poly{n, \Delta}$.
\end{conjecture}

\subsection{Approximate counting}

To estimate the partition function $Z_{G, q}(\beta)$ of anti-ferromagnetic $q$-state Potts model on $\calG_{n, \Delta}^{\BIP}$, we apply the method of abstract polymer models at low temperature.
We put the following conjecture on counting anti-ferromagnetic $q$-Potts models on random regular bipartite graphs.
\begin{conjecture}
    For anti-ferromagnetic $q$-Potts models on random $\Delta$-regular bipartite graphs at $\beta > 0$, with high probability there exists an $\FPTAS$ for $Z_{G, q}(\beta)$.
\end{conjecture}

Meanwhile, the \Erdos-\Renyi random bipartite graph $\calG(n, n, p)$ is another important random models in combinatorics and theoretic computer science.
We generate $G = (\calL, \calR, E) \sim \calG(n, n, p)$ of size $\abs{\calL} = \abs{\calR} = n$ by picking each edge $(u, v) \in \calL \times \calR$ in $E$ with probability $p$ independently at random.
We put the following conjecture on counting anti-ferromagnetic Potts models on sparse \Erdos-\Renyi random bipartite graphs.
\begin{conjecture}
    For anti-ferromagnetic $q$-Potts models on \Erdos-\Renyi random bipartite graph $\calG(n, n, d/n)$ at $\beta > 0$, with high probability there exists an $\FPTAS$ for $Z_{G, q}(\beta)$.
\end{conjecture}
\begin{remark}
    The technique we apply for random regular bipartite graphs seems to fail on the \Erdos-\Renyi random bipartite graphs, since it is seemingly impossible to verify the convergent criterion when the underlying graph is not a bipartite expander.
    Also, it remains a barrier to design the approximating algorithm on unbounded-degree graphs (\Cref{thm:counting-polymer} has a strict restriction on the assumption that the degree $\Delta$ is a constant).
\end{remark}

\bibliographystyle{alpha}
\bibliography{refs}

\end{document}